%% file: Saddlepoint.tex
\documentclass[english,conference,final]{IEEEtran}
\usepackage[T1]{fontenc}
\usepackage[latin9]{inputenc}
\usepackage{amsthm}
\usepackage{amsmath}
\usepackage{amssymb}
\usepackage{graphicx}
\usepackage{esint}
\usepackage{flushend}
\usepackage{cite}  
\usepackage{enumerate}
\usepackage{epstopdf}
\usepackage{color}
 
\makeatletter

\input{preamble} 
 
\theoremstyle{plain}
 
\theoremstyle{plain}
\newtheorem{thm}{\protect\theoremname}
\theoremstyle{plain}
\newtheorem{lem}{\protect\lemmaname} 

\DeclareMathOperator*{\argmax}{arg\,max}

\makeatother

\usepackage{babel}
  \providecommand{\propositionname}{Proposition}
\providecommand{\theoremname}{Theorem}
\providecommand{\lemmaname}{Lemma}

\begin{document}

\title{The Saddlepoint Approximation: Unified Random Coding Asymptotics for Fixed and Varying Rates}

\author{
\authorblockN{Jonathan Scarlett} \authorblockA{University of Cambridge \\ \tt{jms265@cam.ac.uk}} 
\and 
\authorblockN{Alfonso Martinez} \authorblockA{Universitat Pompeu Fabra \\ \tt{alfonso.martinez@ieee.org}} 
\and 
\authorblockN{Albert Guill{\'e}n i F{\`a}bregas} \authorblockA{ ICREA \& Universitat Pompeu Fabra\\ University of Cambridge\\ \tt{guillen@ieee.org}}}

\maketitle
\long\def\symbolfootnote[#1]#2{\begingroup\def\thefootnote{\fnsymbol{footnote}}\footnote[#1]{#2}\endgroup}
\begin{abstract}
    This paper presents a saddlepoint approximation of the random-coding union
    bound of Polyanskiy \emph{et al.} for i.i.d. random coding over discrete memoryless
    channels.  The approximation is single-letter, and can thus be computed efficiently. 
    Moreover, it is shown to be asymptotically tight for both fixed and varying rates, 
    unifying existing achievability results in the
    regimes of error exponents, second-order coding rates, and moderate deviations. 
    For fixed rates, novel exact-asymptotics expressions are specified to within a 
    multiplicative $1+o(1)$ term.  A numerical example is provided for which 
    the approximation is remarkably accurate even at short block lengths.
\end{abstract}
\symbolfootnote[0]{
    This work has been funded in part by the European Research Council under 
    ERC grant agreement 259663, by the European Union's 7th Framework Programme 
    (PEOPLE-2011-CIG) under grant agreement 303633 and by the Spanish Ministry of 
    Economy and Competitiveness under grants RYC-2011-08150 and TEC2012-38800-C03-03.}

\vspace{-0.5cm}    
\section{Introduction} \label{sec:INTRO}

In this paper, we consider problem of channel coding over a discrete memoryless channel $W(y|x)$.
There exists extensive literature studying the tradeoff between the rate $R$, error probability 
$p_e$ and block length $n$, including:
\begin{enumerate}
    \item Error exponents ($R < C$, exponentially decaying $p_e$) \cite{Gallager};
    \item Second-order coding rates ($R \to C$, fixed $p_e$) \cite{Strassen,Finite};
    \item Moderate deviations ($R \to C$ and $p_e \to 0$ simultaneously) \cite{ModerateDev},
\end{enumerate}  
where $C$ is the capacity.
These asymptotic notions provide valuable insight, 
but at finite block lengths it is generally unclear which one
dictates the performance.

In \cite[Sec. III]{Finite}, a non-asymptotic approach was taken. 
The most powerful of the achievability bounds therein is the 
random-coding union (RCU) bound, given by
\vspace{-0.15cm} 
\begin{multline}
    \rcu(n,M) \triangleq \EE\big[\min\big\{1, \\
    (M-1)\PP[W^{n}(\Yv|\Xvbar)\ge W^{n}(\Yv|\Xv)\,|\,\Xv,\Yv]\big\}\big], \label{eq:SU_RCU}
\end{multline}
where $M=e^{nR}$ is the number of messages, $(\Xv,\Yv,\Xvbar) \sim Q^n(\xv)W^n(\yv|\xv)Q^n(\xvbar)$,
$W^{n}(\yv|\xv)\triangleq\prod_{i=1}^{n}W(y_{i}|x_{i})$, and  $Q^n(\xv)\triangleq\prod_{i=1}^{n}Q(x_i)$
for some input distribution $Q$ (here we focus on i.i.d. random coding).
The RCU bound has been shown to be close to non-asymptotic converse bounds 
in several numerical examples \cite{Finite}, but its computation 
is generally prohibitively complex beyond symmetric setups.

In \cite{Saddlepoint}, a saddlepoint approximation \cite{SaddlepointBook} 
was derived for a weakened bound, obtained from \eqref{eq:SU_RCU} using
Markov's inequality:
\begin{equation}
    \rcus(n,M) \triangleq \EE\Big[\min\big\{1,(M-1)e^{-i_s^n(\Xv,\Yv)} \big\}\Big], \label{eq:SU_RCU_s}
\end{equation}
where $s>0$ is arbitrary, and we define the generalized information density 
\begin{align}
    i_s^n(\xv,\yv) &\triangleq \sum_{i=1}^{n}i_s(x_i,y_i) \label{eq:SU_is_n} \\
    i_s(x,y)       &\triangleq \log\frac{W(y|x)^s}{\sum_{\xbar}Q(\xbar)W(y|\xbar)^s}. \label{eq:SU_is}
\end{align}
The approximation in \cite{Saddlepoint} is \emph{single-letter} and takes the form
$\rcushat(n,M) = \alpha_n(Q,R,s)e^{-n\Eriid(Q,R,s)}$,
where $\Eriid$ and $\alpha_n$ represent the error exponent and the subexponential prefactor respectively.
Numerical examples in \cite{Saddlepoint} showed the approximation 
to be remarkably tight, while being essentially as easy to compute
as the exponent alone.  However, its derivation used heuristic arguments.
The techniques of this paper formalize these arguments, and yield
\begin{equation} 
    \lim_{n\to\infty} \frac{\rcushat(n,M_n)}{\rcus(n,M_n)} = 1 \label{eq:SU_SaddleRCUs}
\end{equation}
at both fixed and varying rates. Moreover, both the lattice and 
non-lattice case (see Section \ref{sec:SU_REFINED_IID}) are handled.
Since $\rcus$ can be used to derive the random-coding exponent
\cite[Ch. 5]{Gallager}, channel dispersion \cite{Finite} and moderate
deviations result \cite{ModerateDev}, we conclude from \eqref{eq:SU_SaddleRCUs}
that $\rcushat$ unifies these regimes.

In Theorem \ref{thm:SA_SaddleRCU} below, we present a refined asymptotic bound $\rcuss$
and a corresponding saddlepoint approximation $\rcusshat$ which is tight in the sense of 
\eqref{eq:SU_SaddleRCUs}, and which is seen to approximate the more powerful bound $\rcu$ 
remarkably well numerically (see Figure \ref{fig:SA_IID_R}).
This saddlepoint approximation not only unifies the above-mentioned regimes, but also
characterizes the higher-order asymptotics.  In particular, for a fixed error
probability the approximation captures the third-order $\frac{1}{2}\log n$
term \cite[Sec. 3.4.5]{FiniteThesis}, and for a fixed rate we obtain
the prefactor growth rate derived in \cite{RefinementJournal} (see also
\cite{PaperRefinement}), along with a novel characterization 
of the multiplicative $O(1)$ terms. 

\vspace{-0.15cm}
\section{Preliminary Definitions and Results} \label{sec:SU_PRELIM}

We henceforth make use of the standard asymptotic notations $O(\cdot)$, $o(\cdot)$, 
$\Theta(\cdot)$, $\Omega(\cdot)$ and $\omega(\cdot)$.  

\subsubsection{Information Density Moments and $E_0$ Function}

We write the mean and variance of the information density as
\begin{align}
    I_s(Q)         &\triangleq \EE[i_s(X,Y)] \label{eq:SU_Is} \\
    U_s(Q)         &\triangleq \var[i_s(X,Y)] \label{eq:SU_Us},
\end{align}
where $(X,Y)\sim Q\times W$.  Note that $I_1(Q) = I(X;Y)$.

Following Gallager \cite[Ch. 5]{Gallager}, we define the $E_0$ function
\begin{equation}
    \Eziid(Q,\rho,s) \triangleq -\log\EE\big[e^{-\rho i_s(X,Y)}\big] \label{eq:SU_E0s_IID}
\end{equation} 
and the random-coding error exponent
\begin{equation}
    \Eriid(Q,R) \triangleq \sup_{s>0,\rho\in[0,1]}\Eziid(Q,\rho,s) - \rho R. \label{eq:SU_Exponent}
\end{equation}
While the supremum is achieved by $s=\frac{1}{1+\rho}$ \cite[Ex. 5.6]{Gallager}, 
it will be convenient to consider an arbitrary choice of $s>0$.

The optimal $\rho$ in \eqref{eq:SU_Exponent} for a given value of $s$ is denoted by
\begin{equation}
    \rhohat(Q,R,s) \triangleq \argmax_{\rho\in[0,1]} \Eziid(Q,\rho,s) - \rho R, \label{eq:SU_rho_hat}
\end{equation} 
and the critical rate is defined as
\begin{equation}
    \Rcrs(Q) \triangleq \sup\big\{R\,:\,\rhohat(Q,R,s)=1\big\}.
\end{equation}
We define the following derivatives associated with \eqref{eq:SU_rho_hat}:
\begin{align}
    c_{1}(Q,R,s) & \triangleq R-\frac{\partial \Eziid(Q,\rho,s)}{\partial\rho}\bigg|_{\rho=\rhohat(Q,R,s)}\label{eq:SA_c1} \\
    c_{2}(Q,R,s) & \triangleq-\frac{\partial^{2}\Eziid(Q,\rho,s)}{\partial\rho^{2}}\bigg|_{\rho=\rhohat(Q,R,s)}. \label{eq:SA_c2} 
\end{align}
The following properties of the above quantities coincide with those given by
Gallager \cite[pp. 141-143]{Gallager}, and follow by adapting the arguments
therein to the case of a fixed $s>0$:
\begin{itemize}
  \item If $U_s(Q)>0$, then $c_{2} > 0$ for all $R$;
  \item For $R \in \big[0,\Rcrs(Q)\big)$, we have $\rhohat=1$ and $c_{1}<0$;
  \item For $R \in \big[\Rcrs(Q),I_s(Q)\big]$, $\rhohat$ is strictly decreasing in $R$, and $c_{1}=0$;
  \item For $R > I_s(Q)$, we have $\rhohat=0$ and $c_{1}>0$.
\end{itemize}
Here and throughout the paper, the arguments to $\rhohat$, $c_1$, etc. are
omitted when their values are clear from the context. 

\subsubsection{Singular vs. Non-Singular Case}

Given an input distribution $Q$ and channel $W$, we define the set
\begin{multline}
    \Yc_{1}(Q) \triangleq \Big\{ y\,:\, W(y|x) \ne W(y|\xbar)\text{ for some } x,\xbar \\ 
    \text{ such that }Q(x)Q(\xbar)W(y|x)W(y|\xbar)>0\Big\}. \label{eq:REF_SetY1}
\end{multline}
Following the terminology of Altu\u{g} and Wagner \cite{RefinementJournal}, we say 
that $(Q,W)$ is singular if $\Yc_{1}(Q) = \emptyset$, and non-singular
otherwise.  Our techniques can be used to handle both 
cases.  In the singular case, we in fact have $\rcu=\rcu_s$ \cite{JournalSU}, and hence 
\eqref{eq:SU_SaddleRCUs} gives the desired result regarding
the approximation of $\rcu$.  
In fact, our analysis can be applied directly
to the dependence-testing (DT) bound \cite{Finite}, which improves
(slightly) on $\rcu$ for singular channels.
We focus on the non-singular case, and refer the reader to
\cite{JournalSU} for the singular case.

\subsubsection{Further Definitions}

We say that $Z$ is a lattice random variable with offset $\gamma$ and span $h$ if its 
support is a subset of the lattice $\{\gamma+ih\,:\,i\in\ZZ\}$, and the same cannot remain true by increasing $h$.
Our main result treats two cases separately depending on whether $i_s(X,Y)$ is a lattice variable. 

The density of a $N(\mu,\sigma^2)$ random variable is denoted by
\begin{equation}
    \phi(z;\mu,\sigma^2) \triangleq \frac{1}{\sqrt{2\pi\sigma^2}}e^{-\frac{(z-\mu)^2}{2\sigma^2}}. \label{eq:SU_phi}
\end{equation} 
In the lattice case, we similarly write
\begin{equation}
    \phi_h(z;\mu,\sigma^2) \triangleq \frac{h}{\sqrt{2\pi\sigma^2}}e^{-\frac{(z-\mu)^2}{2\sigma^2}}. \label{eq:SU_phi_h}
\end{equation} 

The remaining definitions are somewhat more technical.
We define the reverse conditional distribution
\begin{align}
    \Ptilde_s(x|y) & \defeq\frac{Q(x)W(y|x)^{s}}{\sum_{\xbar}Q(\xbar)W(y|\xbar)^{s}}, \label{eq:REF_Vs}
\end{align}
the joint tilted distribution  
\begin{equation}
    P_{\rhohat,s}^{*}(x,y)=\frac{Q(x)W(y|x)e^{-\rhohat i_s(x,y)}}{\sum_{x^{\prime},y^{\prime}}Q(x^{\prime})W(y^{\prime}|x^{\prime})e^{-\rhohat i_s(x',y')}}, \label{eq:SU_P*XY}
\end{equation}
and the conditional variance
\begin{equation}
    c_3(Q,R,s) \triangleq \EE\Big[ \var\big[i_s(X^{*}_s,Y_s^{*}) \big| Y_s^{*} \big] \Big], \label{eq:REF_sigma_s}
\end{equation}
where $(X_s^*,Y_s^*) \sim P_{\rhohat,s}^{*}(y)\Ptilde_s(x|y)$, and $P_{\rhohat,s}^{*}(y)$
is the $y$-marginal of \eqref{eq:SU_P*XY}.  We have \cite[Eq. (61)]{PaperRefinement}
\begin{align} 
     \var_{\Ptilde_s(\cdot|y)}[i_{s}(X_s,y)] > 0 \iff y\in\Yc_{1}(Q).\label{eq:REF_VarProof4}
\end{align} 
Furthermore, using \eqref{eq:SU_P*XY}, we have $P_{\rhohat,s}^{*}(y) > 0$ if
and only if $\sum_{x}Q(x)W(y|x)>0$. Combining these, we see that the  
non-singularity assumption implies $c_3 > 0$ for all $R$ and $s>0$.

Finally, we define
\begin{align}
    \Ic_s &\triangleq \Big\{ i_s(x,y) \,:\, Q(x)W(y|x)>0,y\in\Yc_1(Q)\Big\} \\
    \psi_{s} &\triangleq 
        \begin{cases}
            1 & \Ic_s\text{ does not lie on a lattice} \\
            \frac{\hover}{1-e^{-\hover}} & \Ic_s\text{ lies on a lattice with span }\hover.
        \end{cases} \label{eq:SA_Psi_s}
\end{align}

\section{Main Result} \label{sec:SU_REFINED_IID}
 
Our saddlepoint approximation is written in the form
\begin{equation}
    \rcusshat(n,M)\triangleq\beta_n(Q,R,s)e^{-n(\Eziid(Q,\rhohat,s)-\rhohat R)}. \label{eq:SA_RCU_hat}
\end{equation}
We treat the lattice and non-lattice cases separately, defining
\begin{align}
    \hspace*{-0.2cm} \beta_n \triangleq \begin{cases} \betanl & i_s(X,Y)\text{ is non-lattice} \\ \betal & R-i_s(X,Y)\text{ has offset $\gamma$ and span $h$}, \end{cases} \label{eq:SU_Beta_n} 
\end{align}
where
\begin{align}
    & \betanl(Q,R,s) \triangleq \int_{\log\frac{\sqrt{2\pi nc_3}}{\psi_s}}^{\infty}e^{-\rhohat z}\phi(z;nc_1,nc_2)dz \nonumber \\ 
        & + \frac{\psi_s}{\sqrt{2\pi nc_3}}\int_{-\infty}^{\log\frac{\sqrt{2\pi nc_3}}{\psi_s}}e^{(1-\rhohat)z}\phi(z;nc_1,nc_2)dz, \label{eq:SA_PreFactorRCU}
\end{align} 
\begin{align}
    &\betal(Q,R,s) \triangleq \sum_{i=i^{*}}^\infty e^{-\rhohat(\gamma_n+ih)}\phi_h(\gamma_n+ih;nc_1,nc_2) \nonumber \\
        & + \frac{\psi_s}{\sqrt{2\pi nc_3}}\sum_{i=-\infty}^{i^{*}-1}e^{(1-\rhohat)(\gamma_n+ih)}\phi_h(\gamma_n+ih;nc_1,nc_2), \label{eq:SA_PreFactorRCU_L}
\end{align}
and where in \eqref{eq:SA_PreFactorRCU_L} we define
\begin{align}
    \gamma_n &\triangleq \min\Big\{ n\gamma+ih \,:\, i\in\ZZ, n\gamma + ih \ge 0 \Big\} \label{eq:SU_alphan}, \\
    i^{*}    &\triangleq \min\bigg\{ i\in\ZZ \,:\, \gamma_n + ih \ge \log\frac{\sqrt{2\pi nc_3}}{\psi_s} \bigg\}. \label{eq:SU_i*}
\end{align}
While \eqref{eq:SA_PreFactorRCU} and \eqref{eq:SA_PreFactorRCU_L} are written in
terms of integrals and summations, both are single-letter and can be computed 
efficiently, with a complexity which is independent of $n$.  In the non-lattice case, 
this is done by noting that
\begin{equation}
    \int_{a}^{\infty}e^{bz}\phi(z;\mu,\sigma^2)dz = e^{\mu b + \frac{1}{2}\sigma^2b^2}\Qsf\Big(\frac{a-\mu-b\sigma^2}{\sigma}\Big). \label{eq:SU_Integral}
\end{equation}
In the lattice case, we can write each summation in
\eqref{eq:SA_PreFactorRCU_L} as
\begin{equation}
    \sum_i e^{b_0 + b_1i + b_2i^2} = e^{-\frac{b_1^2}{4b_2}+b_0}\sum_i e^{b_2(i+\frac{b_1}{2b_2})^2}, \label{eq:SA_SampleSum} \\
\end{equation}
where $b_2 < 0$.  We can thus obtain an accurate approximation by
keeping only the terms in the summation such that $i$ is sufficiently 
close to $-\frac{b_1}{2b_2}$.  Overall, the computational complexity for any given $s>0$
is similar to that of computing the error exponent alone.  
In principle, the parameter $s$ may be further optimized, but numerical studies indicate
that it suffices to choose $s=\frac{1}{1+\rhohat}$ (i.e. the value maximizing $E_0(Q,\rhohat,s)$).

\begin{thm} \label{thm:SA_SaddleRCU}
    Fix the input distribution $Q$, constant $s>0$, and sequence of positive integers
    $\{M_n\}_{n\ge1}$.  If the pair $(Q,W)$ is non-singular, then 
    \begin{equation}
        \rcu(n,M_n) \le \rcuss(n,M_n)(1+o(1)), \label{eq:SU_PfRCU_1}
    \end{equation}
    where
    \begin{equation}
        \rcuss(n,M) \triangleq \EE\bigg[\min\bigg\{1,\frac{M\psi_s}{\sqrt{2\pi nc_3}}e^{-i_s^n(\Xv,\Yv)} \bigg\}\bigg]. \label{eq:SU_RCU_s_star}
    \end{equation}
    Furthermore, we have
    \begin{equation}
        \lim_{n\to\infty} \frac{\rcusshat(n,M_n)}{\rcuss(n,M_n)} = 1. \label{eq:SU_SaddleRCU}
    \end{equation} 
\end{thm}
\begin{IEEEproof}
    See Section \ref{sub:SA_RCU_PROOF}.
\end{IEEEproof}

The proof of Theorem \ref{thm:SA_SaddleRCU} reveals that
for a fixed target error probability we have
$\rcusshat = \rcuss + O\big(\frac{1}{\sqrt n}\big)$.  From the analysis given 
in \cite[Sec. 3.4.5]{FiniteThesis}, setting $\rcuss = \epsilon$ and solving
for the required number of messages yields
\begin{equation}
    \log M = nI_s(Q) - \sqrt{nU_s(Q)}\Qsf^{-1}(\epsilon) + \frac{1}{2}\log n + O(1). \label{eq:SU_ThirdOrder}
\end{equation}
By Taylor expanding the $\Qsf^{-1}$ function, we conclude that
the same is true of $\rcusshat$.  Note that since $I_1(Q) = I(X;Y)$, 
\eqref{eq:SU_ThirdOrder} is primarily of interest when $s=1$ and $Q$ achieves capacity.

For a fixed rate $R \ge 0$, we can apply asymptotic expansions to 
\eqref{eq:SA_PreFactorRCU}--\eqref{eq:SA_PreFactorRCU_L} to show the following 
\cite{JournalSU} (here $f_n \asymp g_n$ means that $\lim_{n\to\infty}\frac{f_n}{g_n}=1$):
\begin{itemize}
  \item If $R\in[0,\Rcrs(Q))$, then $\beta_n(Q,R,s) \asymp \frac{\psi_s}{\sqrt{2\pi nc_3}}$.
  \item If $R=\Rcrs(Q)$, then $\beta_n(Q,R,s) \asymp \frac{\psi_s}{2\sqrt{2\pi nc_3}}$.
  \item If $R\in(\Rcrs(Q),I_s(Q))$, then
        \begin{align}
            &\betanl(Q,R,s)  \asymp \bigg(\frac{\psi_s}{\sqrt{2\pi nc_3}}\bigg)^{\rhohat} \frac{1}{\sqrt{2\pi nc_2}\rhohat(1-\rhohat)}, \label{eq:SU_AboveRcr2} \\
            &\betal(Q,R,s) \asymp \bigg(\frac{\psi_s}{\sqrt{2\pi nc_3}}\bigg)^{\rhohat} \frac{h}{\sqrt{2\pi nc_2}} \nonumber \\
                & \times \Bigg(e^{-\rhohat\gamma'_n}\bigg(\frac{1}{1-e^{-\rhohat h}}\bigg) + e^{(1-\rhohat)\gamma'_n}\bigg(\frac{e^{-(1-\rhohat)h}}{1-e^{-(1-\rhohat)h}}\bigg)\Bigg), \label{eq:SU_AboveRcr_L2}
        \end{align}
        where $\gamma'_n \triangleq \gamma_n + i^{*}\,h - \log\frac{\sqrt{2\pi nc_3}}{\psi_s} \in [0,h)$ (see \eqref{eq:SU_i*}).
  \item If $R=I_s(Q)$, then $\beta_n(Q,R,s) \asymp \frac{1}{2}$.
  \item If $R>I_s(Q)$, then $\beta_n(Q,R,s) \asymp 1$.
\end{itemize} 
When combined with Theorem \ref{thm:SA_SaddleRCU}, these expansions 
provide an alternative proof of the main result of Altu\u{g}-Wagner \cite{RefinementJournal},
and an explicit characterization of the multiplicative $O(1)$ terms. 

\subsection{Numerical Example} \label{sub:SA_NUMERICAL}
 
A numerical example is given in Figure \ref{fig:SA_IID_R} (see the caption
for details).  Definitions of the error exponent and normal approximations
can be found in \cite{Finite}, and the exact asymptotics approximation
equals the right-hand side of \eqref{eq:SU_AboveRcr_L2}.
We set $s=1$ for the normal approximation, and $s=\frac{1}{1+\rhohat}$ 
for the other approximations.

We see that the saddlepoint approximation provides an excellent approximation of
$\rcu(n,M)$.   The exact asymptotics approximation is
accurate other than a divergence near the critical rate.  A similar divergence
also occurs near capacity, but this is not visible in the plot; see \cite{JournalSU} 
for further discussion.  In this example, neither the error exponent approximation
nor normal approximation is accurate, though the latter moves closer 
to $\rcu$ upon including the $\frac{1}{2}\log n$ term.
Roughly speaking, the normal (respectively, error exponent) approximation
is better suited to rates near capacity (respectively, low rates), whereas 
the saddlepoint approximation is accurate at all rates.

It should be noted that the observed accuracy of the saddlepoint approximation
is not limited to symmetric setups; see \cite{Saddlepoint,JournalSU} for
further examples.

\begin{figure}
    \begin{centering} 
        \includegraphics[width=0.95\columnwidth]{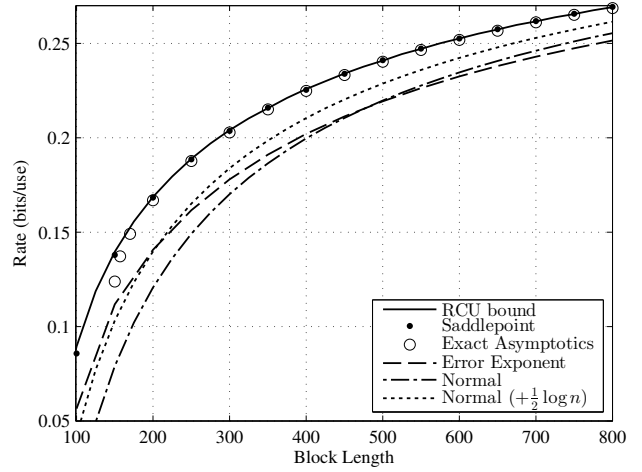}
        \par 
    \end{centering}
 
    \vspace{-3mm}
    \caption{Rate required to achieve a target error probability $\epsilon=10^{-5}$ for
             the binary symmetric channel with crossover probability $\delta=0.15$,
             and the uniform input distribution $Q=(\frac{1}{2},\frac{1}{2})$.
             This corresponds to the lattice case in \eqref{eq:SU_Beta_n}.
             The capacity and critical rate are 0.390 bits/use and 0.124 bits/use.
             \vspace{-3mm}
        \label{fig:SA_IID_R}}
\end{figure}

\section{Proof of Theorem \ref{thm:SA_SaddleRCU}} \label{sub:SU_SADDLEPOINT_PROOF}

Due to space constraints, we omit some details and focus on the non-lattice
case. Full details can be found in \cite{JournalSU}.

\subsection{Proof of \eqref{eq:SU_SaddleRCU}} 

\subsubsection{Alternative Expressions for $\rcuss$} \label{sub:SA_ABOVE_RCR}

For any non-negative random variable $A$, we have $\EE[\min\{1,A\}]=\PP[A \ge U]$,
where $U$ is uniform on $(0,1)$ and independent of $A$.  Defining
$g_n \triangleq \frac{1}{\psi_s}\sqrt{2\pi nc_3}$, we can thus write
\eqref{eq:SU_RCU_s_star} as
\begin{equation}
    \rcuss(n,M)=\PP\bigg[nR-\sum_{i=1}^{n}i_{s}(X_{i},Y_{i})\ge\log(Ug_n)\bigg]. \label{eq:SA_Above_1}
\end{equation} 
Let $F(t)$ denote the cumulative distribution function (CDF) of $R-i_{s}(X,Y)$ 
and let $Z_{1},\cdots,Z_{n}$ be i.i.d. with CDF
\begin{equation}
    F_{Z}(z)=e^{\Eziid-\rhohat R}\int_{-\infty}^{z}e^{\rhohat t}dF(t), \label{eq:SA_Above_2}
\end{equation}
where the arguments to $\Eziid$ are kept implicit.
Using a standard change of measure argument, we showed in \cite[Eq. (44)]{PaperRefinement} that
\begin{equation}
    \rcuss(n,M) = I_n e^{-n(\Eziid(Q,\rhohat,s)-\rhohat R)}, \label{eq:SA_Above_3} 
\end{equation}
where
\begin{equation}
    I_n \triangleq \int_0^1\int_{\log(ug_n)}^\infty e^{-\rhohat z} dF_{n}(z) dF_{U}(u), \label{eq:SA_ChgMeasure2}
\end{equation}   
and where $F_n$ is the CDF of $\sum_{i=1}^n Z_i$, and $F_U$ is the CDF of U.  
Moreover, we showed in \cite[Eqs. (48)--(49)]{PaperRefinement} that 
\begin{equation}
    \EE[Z] = c_1, \quad \var[Z] = c_2, \label{eq:SU_MGFder1}
\end{equation}
where $c_1$ and $c_2$ are defined in \eqref{eq:SA_c1}--\eqref{eq:SA_c2}.
It is not difficult to show that the non-singularity assumption
implies $U_s(Q)>0$, which in turn implies $c_2>0$ (see Section \ref{sec:SU_PRELIM}).

Since the integrand in \eqref{eq:SA_ChgMeasure2} is non-negative, we can 
safely interchange the order of integration, yielding 
\begin{align}
    I_n &= \int_{-\infty}^{\infty} \int_0^{\min\big\{1,\frac{1}{g_n}e^z\big\}} e^{-\rhohat z} dF_U(u) dF_n(z) \label{eq:SA_ChgMeasure4} \\
        &= \int_{\log g_n}^{\infty} e^{-\rhohat z} dF_n(z) + \frac{1}{g_n}\int_{-\infty}^{\log g_n} e^{(1-\rhohat)z} dF_n(z), \label{eq:SA_ChgMeasure5a} 
\end{align} 
where \eqref{eq:SA_ChgMeasure5a} follows by splitting the integral according
to which value achieves the $\min\{\cdot,\cdot\}$ in \eqref{eq:SA_ChgMeasure4}.
Letting $\hat{F}_n$ denote the CDF of $\frac{\sum_{i=1}^n Z_i - nc_1}{\sqrt{nc_2}}$,
we can write \eqref{eq:SA_ChgMeasure5a} as
\begin{multline}
     I_n = \int_{\frac{\log g_n-nc_1}{\sqrt{nc_2}}}^{\infty} e^{-\rhohat(z\sqrt{nc_2}+nc_1)} d\hat{F}_{n}(z) \\ 
        + \frac{1}{g_n}\int_{-\infty}^{\frac{\log g_n-nc_1}{\sqrt{nc_2}}} e^{(1-\rhohat)(z\sqrt{nc_2}+nc_1)} d\hat{F}_{n}(z). \label{eq:SA_ChgMeasure5} 
\end{multline} 
 
\subsubsection{Application of a Refined Central Limit Theorem}

Let $\Phi(z)$ denote the CDF of a zero-mean unit-variance Gaussian
random variable. Using the fact that $\EE[Z]=c_1$ and $\var[Z]=c_2>0$, we have from 
the refined central limit theorem in \cite[Sec. XVI.4, Thm. 1]{Feller} that
\begin{equation}
    \hat{F}_{n}(z)=\Phi(z)+G_{n}(z)+\tilde{F}_{n}(z),\label{eq:SA_Above_6}
\end{equation}
where $\tilde{F}_{n}(z)=o(n^{-\frac{1}{2}})$ uniformly in $z$, and
\begin{equation}
    G_{n}(z)\triangleq\frac{K}{\sqrt{n}}(1-z^{2})e^{-\frac{1}{2}z^{2}}\label{eq:SA_Above_7}
\end{equation}
for some constant $K$ depending only on the variance and third absolute moment
of $Z$. Substituting \eqref{eq:SA_Above_6} into \eqref{eq:SA_ChgMeasure5}, we obtain
\begin{equation}
    I_{n}=I_{1,n}+I_{2,n}+I_{3,n},\label{eq:SA_Above_8} 
\end{equation}
where the three terms denote the right-hand side of \eqref{eq:SA_ChgMeasure5}
with $\Phi$, $G_n$ and $\tilde{F}_n$ respectively in place of $\hat{F}_n$.  
By reversing the step from \eqref{eq:SA_ChgMeasure5a} to \eqref{eq:SA_ChgMeasure5},
we see that $I_{1,n}$ is precisely $\betanl$ in 
\eqref{eq:SA_PreFactorRCU}. In accordance with the theorem statement, we must show that 
$I_{2,n}=o(\betanl)$ and $I_{3,n}=o(\betanl)$ even 
when $R$ and $\rhohat$ vary with $n$.  Let $R_n \triangleq \frac{1}{n}\log M_n$
and $\rhohat_n\triangleq\rhohat(Q,R_n,s)$, and let $c_{1,n}$ and $c_{2,n}$ be the corresponding
values of $c_1$ and $c_2$.  We assume with no real loss of generality that
\begin{align}
    & \lim_{n\to\infty}R_n = R^{*} \label{eq:SA_Rstar}
\end{align}  
for some $R^{*}\ge0$ possibly equal to $\infty$.  Once \eqref{eq:SU_SaddleRCU}
is proved for all such $R^{*}$, the same will follow for
arbitrary $\{R_n\}$. 

\begin{table*}
    \caption{Growth rates of $\betanl$, $I_{2,n}$ and $I_{3,n}$ when the rate
    converges to $R^{*}$. \label{tab:SU_rcusGrowth}}
    \begin{centering}
    \begin{tabular}{|c|c|c|c|c|c|c|}
    \hline 
     & $\hat{\rho}$ & $c_{1}$ & Dominant Term(s) & $\betanl$ & $I_{2,n}$ & $I_{3,n}$\tabularnewline
    \hline 
    \hline 
    $R^{*}\in[0,\Rcrs(Q))$ & $1$ & $<0$ & 2 & $\Theta\Big(\frac{1}{\sqrt n}\Big)$ &  $\Theta\Big(\frac{1}{n}\Big)$ & $o\Big(\frac{1}{n}\Big)$\tabularnewline
    \hline 
    $R^{*}=\Rcrs(Q)$ & $\to1$ & $\to0$ & 2 & $\omega\Big(\frac{1}{n}\Big)$ &  $O\Big(\frac{1}{n}\Big)$ & $o\Big(\frac{1}{n}\Big)$\tabularnewline
    \hline 
    $R^{*}\in(\Rcrs(Q),I_{s}(Q))$ & $\in(0,1)$ & $0$ & 1,2 & $\Theta\Big(\frac{1}{n^{\frac{1}{2}(1+\rhohat)}}\Big)$ & $\Theta\Big(\frac{1}{n^{\frac{1}{2}(2+\rhohat)}}\Big)$ & $o\Big(\frac{1}{n^{\frac{1}{2}(1+\rhohat)}}\Big)$ \tabularnewline
    \hline 
    $R^{*}=I_{s}(Q)$ & $\to0$ & $\to0$ & 1 & $\omega\Big(\frac{1}{\sqrt{n}}\Big)$ & $O\Big(\frac{1}{\sqrt{n}}\Big)$ & $o\Big(\frac{1}{\sqrt{n}}\Big)$\tabularnewline
    \hline 
    $R^{*}>I_{s}(Q)$ & $0$ & $>0$ & 1 & $\Theta(1)$ & $\Theta\Big(\frac{1}{\sqrt{n}}\Big)$ & $o\Big(\frac{1}{\sqrt{n}}\Big)$\tabularnewline
    \hline 
    \end{tabular}
    \par\end{centering}
\end{table*}

Table \ref{tab:SU_rcusGrowth} summarizes the growth rates $\betanl$, $I_{2,n}$ and
$I_{3,n}$ for various ranges of $R^*$, and indicates whether the first
or second integral (see \eqref{eq:SA_ChgMeasure5}) dominates the behavior
of each.  We see that $I_{2,n}=o(\betanl)$ and $I_{3,n}=o(\betanl)$ for
all $R^{*}$, as desired.
 
As an example, we consider the case $R^{*}\in(\Rcrs(Q),I_{s}(Q))$.  
The given behavior of $\betanl$ follows immediately from \eqref{eq:SU_AboveRcr2}.
Taking the derivative of $G_n(z)$ in \eqref{eq:SA_Above_7},
we can evaluate $I_{2,n}$ by writing it in terms of the standard Gaussian  
density $\phi(z)=\frac{1}{\sqrt{2\pi}}=e^{-z^2/2}$.  
For $I_{3,n}$, we analyze the two integrals in a similar fashion;
here we focus on the first. For the integration range given,
the integrand is upper bounded by $e^{-\rhohat \log g_n} = \Theta(n^{-\rhohat/2})$. 
Combining this with the fact that $\tilde{F}_{n}(z)=o(n^{-\frac{1}{2}})$ uniformly in $z$,
we obtain the desired $o(n^{-\frac{1}{2}(1+\rhohat)})$ decay rate.

\subsection{Proof of \eqref{eq:SU_PfRCU_1}} \label{sub:SA_RCU_PROOF}

To prove \eqref{eq:SU_PfRCU_1}, we make use of two technical lemmas, whose 
proofs can be found in \cite[Appendix F]{JournalSU}.

\begin{lem} \label{lem:REF_Lem20}
    Fix $K>0$, and for each $n$, let $(n_1,\cdots,n_K)$ be integers such that $\sum_{k}n_k=n$.
    Fix the probability mass functions (PMFs) $Q_1,\cdots,Q_K$ on
    a common finite alphabet, and let $\sigma_1^2,\cdots,\sigma_K^2$ be the corresponding
    variances.  Let $Z_{1},\cdots,Z_{n}$ be independent random variables, $n_k$ of which
    are distributed according to $Q_k$ for each $k$.  Suppose that 
    $\min_{k}\sigma_k > 0$ and $\min_{k}n_k = \Theta(n)$.  
    Defining
    \begin{align}
        \Ic_0  &\triangleq \bigcup_{k \,:\, \sigma_{k} > 0}\big\{ z \,:\, Q_k(z)>0 \big\} \label{eq:SA_I0} \\
        \psi_0 &\triangleq 
            \begin{cases}
                1 & \Ic_0\text{ \em does not lie on a lattice} \\
                \frac{h_0}{1-e^{-h_0}} & \Ic_0\text{ \em lies on a lattice with span }h_0,
            \end{cases} \label{eq:SA_Psi_0}
    \end{align}
    the sum $S_n\triangleq\sum_{i}Z_i$ satisfies the 
    following uniformly in $t$:
    \begin{equation}
        \EE\Big[e^{-S_n}\emph{\openone}\big\{S_n \ge t\big\}\Big] \le e^{-t}\bigg(\frac{\psi_0}{\sqrt{2\pi V_n}} + o\Big(\frac{1}{\sqrt{n}}\Big)\bigg), \label{eq:SU_Lemma20}
    \end{equation}
    where $V_n \triangleq \var[S_n]$, and $\emph{\openone}\{\cdot\}$ is the indicator function.
\end{lem}
\begin{proof}
    The proof is analogous to that of \cite[Lemma 47]{Finite}, except that the use
    of the Berry-Esseen theorem is replaced by the local limit theorems in \cite[Thm. 1]{LLTNonLattice}
    and \cite[Sec. VII.1, Thm. 2]{PetrovBook} for the non-lattice and lattice 
    cases respectively.
\end{proof}

Define the random variables
\begin{equation}
(\Xv,\Yv,\Xvbar,\Xv_{s}) \sim Q^{n}(\xv)W^{n}(\yv|\xv)Q^{n}(\xvbar)\Ptilde_s^{n}(\xv_{s}|\yv), \label{eq:REF_BoldVars}
\end{equation}
where $\Ptilde_s^{n}(\xv|\yv) \triangleq \prod_{i=1}^n \Ptilde_s(x_i|y_i)$.
We write the empirical distribution of $\yv$ as $\hat{P}_{\yv}$,
and we let $P_{\Yv}$ denote the PMF of $\Yv$.
\begin{lem} \label{lem:REF_SetF}
    Let $s>0$ and $\rhohat\in[0,1]$ be given. If the pair $(Q,W)$ is non-singular, then the set
    \begin{equation}
        \Fc_{\rhohat,s}^n(\delta) \defeq \Big\{\yv\,:\,P_{\Yv}(\yv)>0, \, \max_{y}\big|\hat{P}_{\yv}(y)-P_{\rhohat,s}^{*}(y)\big| \le \delta\Big\} \label{eq:REF_SetFn}
    \end{equation}
    satisfies the following properties:
    \begin{enumerate}
        \item For any $\yv\in\Fc_{\rhohat,s}^n(\delta)$, we have
        \begin{equation}
            \var\big[i_{s}^{n}(\Xv_{s},\Yv)\,|\,\Yv=\yv\big]\ge n(c_3 - r(\delta)), \label{eq:REF_VarLB}
        \end{equation}
        where $r(\delta)\to0$ as $\delta\to0$.
        \item For any $\delta>0$, we have
        \begin{align}
            \hspace*{-0.5cm}\liminf_{n\to\infty}-\frac{1}{n}\log\frac{\sum_{\xv,\yv\notin\Fc_{\rhohat,s}^n(\delta)}Q^{n}(\xv)W^{n}(\yv|\xv)e^{-\rhohat i_{s}^{n}(\xv,\yv)}}{e^{-nE_0(Q,\rhohat,s)}} > 0. \label{eq:REF_LimRatio}
        \end{align}
    \end{enumerate}
\end{lem}
\begin{proof}
    This is a simple refinement of \cite[Lemma 3]{PaperRefinement}.
\end{proof}

Since the two statements of Lemma \ref{lem:REF_SetF} hold true
for any $\rhohat\in[0,1]$, they also hold true when $\rhohat$ varies within this range,
thus allowing us to handle rates which vary with $n$.

By upper bounding $M-1$ by $M$ in \eqref{eq:SU_RCU}, we obtain
\begin{align}
    \rcu(n,M) \le S_0(\rhohat,s,\delta) + \sum_{\xv,\yv\in\Fc_{\rhohat,s}^n(\delta)}Q^n(\xv)W^{n}(\yv|\xv) \nonumber \\
        \times \min\Big\{1,M\PP[i_{s}^{n}(\Xvbar,\yv)\ge i_{s}^{n}(\xv,\yv)]\Big\}, \label{eq:REF_ExpandedRCU}
\end{align} 
where $S_0(\rhohat,s,\delta)$ is a sum of the same form
as the second term in \eqref{eq:REF_ExpandedRCU} with $\yv\notin\Fc_{\rhohat,s}^n(\delta)$,
and we have replaced $W^n$ by $i_s^n$ since each is an increasing
function of the other.  Following \cite[Sec. 3.4.5]{FiniteThesis}, 
we have the following when $\Ptilde^n_s(\xvbar,\yv)\ne0$:
\begin{align}
    Q^{n}(\xvbar) & =Q^{n}(\xvbar)\frac{\Ptilde_s^{n}(\xvbar|\yv)}{\Ptilde_s^{n}(\xvbar|\yv)} = \Ptilde_s^{n}(\xvbar|\yv)e^{-i_{s}^{n}(\xvbar,\yv)}.\label{eq:iidDeriv0}
\end{align}
Summing \eqref{eq:iidDeriv0} over all $\xvbar$ such that $i_{s}^{n}(\xvbar,\yv)\ge t$ yields
\begin{equation}
    \PP[i_{s}^{n}(\Xvbar,\yv)\ge t]
    =\EE\Big[e^{-i_{s}^{n}(\Xv_{s},\Yv)}\openone\big\{ i_{s}^{n}(\Xv_{s},\Yv)\ge t\big\}\,\Big|\,\Yv=\yv\Big]\label{eq:iidDeriv1}
\end{equation} 
under the joint distribution in \eqref{eq:REF_BoldVars}. 

We now observe that \eqref{eq:iidDeriv1} is of the same form as
the left-hand side of \eqref{eq:SU_Lemma20}.  We apply Lemma 
\ref{lem:REF_Lem20} with $Q_k$ given by the PMFs of $i_s(X_s,y)$
under $X_s \sim \Ptilde_s(\,\cdot\,|y)$ for the various $y$ values. 
We have from \eqref{eq:SU_Lemma20}, \eqref{eq:REF_VarLB} and \eqref{eq:iidDeriv1} that
\begin{equation}
    \PP\big[i_{s}^{n}(\Xvbar,\yv)\ge t\big] \le \frac{\psi_s}{\sqrt{2\pi n(c_3-r(\delta))}}e^{-t}(1+o(1)) \label{eq:REF_Tail1}
\end{equation}
for all $\yv\in\Fc_{\rhohat,s}^n(\delta)$ and sufficiently small $\delta$
(recall that $c_3 > 0$).
Here we have used the fact that $\psi_0$ in \eqref{eq:SA_Psi_0}
coincides with $\psi_s$ in \eqref{eq:SA_Psi_s}, which follows 
from \eqref{eq:REF_VarProof4} and the fact that $\Ptilde_s(x|y)>0$ if and
only if $Q(x)W(y|x)>0$ (see \eqref{eq:REF_Vs}).

Using the uniformity of the $o(1)$ term in $t$ in \eqref{eq:REF_Tail1} 
(see Lemma \ref{lem:REF_Lem20}), taking $\delta\to0$ (and 
hence $r(\delta)\to0$), and writing
\begin{equation}
    \min\{1,f_n(1+\zeta_n)\} \le (1 + |\zeta_n|)\min\{1,f_n\} \label{eq:SA_min1property},
\end{equation}  
we see that the second term in \eqref{eq:REF_ExpandedRCU} is upper bounded by 
$\rcuss(n,M)(1+o(1))$.  Finally, using \eqref{eq:REF_LimRatio} 
(along with \eqref{eq:SA_RCU_hat} and \eqref{eq:SU_SaddleRCU}), 
it is easily shown that $S_0(\rhohat,s,\delta)$ can be factored 
into the $1+o(1)$ term, thus completing the proof of \eqref{eq:SU_PfRCU_1}.

\bibliographystyle{IEEEtran}
\bibliography{12-Paper,18-MultiUser,18-SingleUser,35-Other}

\end{document}

%% file: preamble.tex
\newcommand{\openone}{\leavevmode\hbox{\small1\normalsize\kern-.33em1}}

\newcommand{\Eziid}{E_{0}}

\newcommand{\Eriid}{E_{r}}

\newcommand{\rcu}{\mathrm{rcu}}
\newcommand{\rcus}{\rcu_{s}}
\newcommand{\rcuss}{\rcu_{s}^{*}}
\newcommand{\rcushat}{\widehat{\mathrm{rcu}}_{s}}
\newcommand{\rcusshat}{\widehat{\mathrm{rcu}}_{s}^{*}}

\newcommand{\rhohat}{\hat{\rho}}

\newcommand{\Rcrs}{R_s^{\mathrm{cr}}}

\newcommand{\Ptilde}{\widetilde{P}}

\newcommand{\hover}{\overline{h}}

\newcommand{\xbar}{\overline{x}}

\newcommand{\xvbar}{\overline{\boldsymbol{x}}}

\newcommand{\xv}{\boldsymbol{x}}

\newcommand{\Xvbar}{\overline{\boldsymbol{X}}}

\newcommand{\Xv}{\boldsymbol{X}}
\newcommand{\yv}{\boldsymbol{y}}
\newcommand{\Yv}{\boldsymbol{Y}}

\newcommand{\Fc}{\mathcal{F}}

\newcommand{\Ic}{\mathcal{I}}

\newcommand{\Yc}{\mathcal{Y}}

\newcommand{\EE}{\mathbb{E}}
\newcommand{\PP}{\mathbb{P}}
\newcommand{\ZZ}{\mathbb{Z}}
\newcommand{\Qsf}{\mathsf{Q}}

\newcommand{\defeq}{\triangleq}
\newcommand{\var}{\mathrm{Var}}

\newcommand{\betanl}{\beta^{\mathrm{nl}}_n}
\newcommand{\betal}{\beta^{\mathrm{l}}_n}